\documentclass{article}
\usepackage{graphicx}
\usepackage{amsmath}
\usepackage{amsfonts}
\usepackage{longtable}
\usepackage{url}
\usepackage{amsmath}
\usepackage{amsfonts}
\usepackage{amsthm}
\usepackage{subcaption}

\def\bbR{\mathbb{R}}
\def\bbZ{\mathbb{Z}}
\def\bbH{\mathbb{H}}

\def\ra{\rightarrow}

\newcounter{mycount}[section]

\newtheorem{proposition}[mycount]{Proposition}

\def\shortonly#1{}
\def\bbX{\mathbb{X}}
\def\bbF{\mathbb{F}}
\def\bbH{\mathbb{H}}

\begin{document}
\title{Generating Regular Hyperbolic Honeycombs}
\author{Dorota Celi\'nska-Kopczy\'nska, Eryk Kopczy\'nski
\\ Institute of Informatics, University of Warsaw
\\ erykk@mimuw.edu.pl
}
%
%\institute{}
\maketitle              % typeset the header of the contribution
\begin{abstract}
Geodesic regular tree structures are essential 
to combat numerical precision issues that arise while working with large-scale computational hyperbolic geometry
and have applications in algorithms based on distances in such tessellations. 
We present a method of generating and applying such structures to the tessellations of 3-dimensional hyperbolic space.
\end{abstract}

Hyperbolic geometry is difficult to work with computationally due to its exponential growth. The volume of a ball of radius $r$ is exponential in $r$,
and thus, any coordinate system based on a tuple of real values will quickly run into numerical issues.
This is a serious problem in applications \cite{reptradeoff,tobias_alenex,rtvizfinal}, and one solution to this problem is to use tessellations, 
which can be generated combinatorially \cite{wpigroups,dhrg_sea,mltiles}.

In \cite{gentes} a general algorithm for generating 2D hyperbolic tessellations is presented. 
This paper presents a conceptually similar method for 3D tessellations (honeycombs). 
The most obvious honeycombs are the \emph{regular} honeycombs, denoted by the Schl\"afli symbol $\{p,q,r\}$, which means that every cell is a Platonic polyhedron with $p$-sided
faces, $q$ faces meeting at every vertex, and $r$ cells meeting at every edge; we recommend the visualizations in the paper by Nelson and Segerman \cite{hhoney}. 
Other than the applications in modeling hierarchical data \cite{bogu_internet,reptradeoff,mltiles}, this issue is important for video games based on
hyperbolic geometry; open-world procedural generation is a very good match for hyperbolic geometry \cite{hyperrogue}, and a good method for representing honeycombs
combinatorially is essential. In our experience, generating the basic hyperbolic honeycombs (such as \{4,3,5\}) is significantly harder than generating
the basic 2D hyperbolic tessellations such as {7,3} or {5,4}. At the moment of writing this paper, we know three projects of such games: HyperRogue, Hypermine and Hyperblock. HyperRogue uses the methods described in this paper.
Hypermine by Ralith uses an approach based on Coxeter triangles and Knuth-Bendix procedure \cite{wpigroups}; as in the 2D case \cite{gentes}, we prefer a method
working with the cells of the intended honeycomb directly. Up to our knowledge, the development of Hyperblock by Kayturs is currently stopped due to the challenge of working with hyperbolic honeycombs.%
\footnote{From all three projects, it seems that regular hyperbolic honeycombs are challenging not only to generate, but also from the purpose of designing levels. It might be easier to use honeycombs based on
the 3D analog of the binary tiling instead. Such an approach was suggested by Henry Segerman.} Another application of tree structures for hyperbolic honeycombs is finding the coordination sequences
($c_n$ is the number of cells $n$ steps from the chosen central cell). While such sequences for the most important regular hyperbolic honeycombs are of interest, prior to our work OEIS included only 
an incorrect coordination sequence for \{4,3,5\}.

We assume the reader is familiar with the basic ideas of the two-dimensional case \cite{gentes}. The main differences in the three-dimensional case are as follows.

\begin{itemize}
\item Two-dimensional tessellations are easily constructed synthetically: we find a collection of tile types, and create rules describing how they are to be connected.
It is enough for the angles at every vertex to sum to $360^\circ$ for this to yield a valid tessellation. This approach has been used to create the most tessellations in
the tessellations catalog \cite{tescatalog}. Up to our knowledge, similar methods have not yet been developed in higher dimensions; obvious honeycombs are obtained as
regular honeycombs, by subdividing the regular honeycombs, or by subdividing the cells of a hyperbolic manifold and taking an universal cover.

\item The synthetic method described above tends to yield tiles with $s$-fold rotational symmetry. The tessellation descriptions used by the algorithm in \cite{gentes} have this symmetry,
which introduces a bit of subtlety when we need to determine the parent of a given tile in the tree structure. In the 3D case, we have instead decided to break the symmetry, by
fixing a specific canonical rotation of every cell. The cells could have $s$-fold symmetries around an axis,
or the symmetries corresponding to one of the Platonic solid; however, all the cases known to us can be handled by finding closed hyperbolic manifolds constructed from gluing 
a number of $\{p,q,r\}$ cells. Once such manifolds are known, we can assume that our tessellation is an universal cover of some subdivision of a hyperbolic manifold, which is enough to
break the symmetry.

\item The method in \cite{gentes} classified the non-tree edges as Left ($L$) and Right ($R$); such a classification was enough to generate the non-tree connection in a very simple
way, and also to verify whether the GRTS learned by the Angluin-style algorithm was correct. In the three-dimensional case the situation is more challenging.

\item If $\{q,r\}$ is a Euclidean tessellation, the honeycomb $\{p,q,r\}$ has ideal vertices \cite{hhoney}; the structure of the cells incident to such an ideal vertex is Euclidean
(the cells roughly form a horosphere), which makes the whole graph of the honeycomb not Gromov hyperbolic. In our experience, this does not seem to affect the generation of honeycombs,
but it makes the amortized time complexity of our generation method, and other GRTS-based algorithms, larger than $O(1)$.

\item Higher dimension generally means more cells in the neighborhood of a given cell, which increases the computational resources required to work with hyperbolic honeycombs. It takes
the current implementation of our algorithm a few hours and over 50 GB of RAM to find.
\end{itemize}

To learn how to generate a hyperbolic honeycomb from a 3D GRTS, the reader has to read the Sections \ref{sec:periodic} and \ref{sec:geodesic}, which explain the definitions.
The descriptions of our periodic honeycombs and their GRTS can be downloaded from \url{https://figshare.com/articles/software/Generating_tree_structures_for_hyperbolic_honeycombs/20713306}.

The remaining sections detail our method of creating a 3D GRTS. Section \ref{sec:breaking} explains how we deal with the symmetries of regular hyperbolic honeycombs, and Section \ref{sec:main}
describes the our algorithm. We conclude with Section \ref{sec:concl}, which summarizes our results.
Section \ref{sec:details} discusses other details of our implementation. Appendix \ref{sec:tested} lists the honeycombs we have
tested our algorithm on and provides the characteristics of the tested honeycombs.

Our methods do not depend on the dimension in any significant way, so they could be potentially also applied to higher dimensions.

\section{Fixed Periodic Honeycombs}\label{sec:periodic}

We use a definition of a periodic honeycomb similar to that of a periodic tessellation in \cite{gentes}. As mentioned in the Introduction, one difference is that we do not allow the tiles
to have their own symmetries, in other words, all the cells have a fixed canonical rotation.

A fixed periodic honeycomb $G$ in space $\mathbb{X}$ (usually $\mathbb{X}=\bbH^3$) has a finite set of cell types $T$. Every cell type $t\in T$ has a polyhedral shape $S_t \in \bbX$ (which can have ideal or ultra-ideal vertices
\cite{hhoney};
however, we assume in this paper that it has only finitely many faces).
Let $F(t)$ be the set of all faces of $S_t$. We denote the set of all faces by $U=\bigcup F_(t)$. 
Every face $f \in F(t)$ connects to some face $f'$ of some $S_{t'}$; we denote this with $(t,f) \sim (t',f')$. When the face $f'$ is irrelevant, we also use the notation $t'=\nu(t,f)$.
There is an orientation-preserving isometry $C_{t,f}$ such that $S_t$ and $C_{t,f}S_{t'}$
share the face $f=C_{t,f}f'$. Every cell $g \in G$ has a type $t_g \in T$, such that there is an orientation-preserving isometry $I_g$ of $\bbX$ such that $g$ (as a subset of $\bbX$) is of form $I_g(S_t)$,
and the cells adjacent to $g$ are of form $I_gC_{t_g,f}S_t'$, where $f$ goes over all faces of $S_t$ and $t'=\nu(t,f)$. We denote the neighbor of $g$ on the other side of the face corresponding to 
$f \in F(t_g)$ by $\nu(g,f)$. These rules imply that if $g_1$ and $g_2$ are of the same type $t$, the isometry
$I_{g_2}I_{g_1}^{-1}$ is an isometry of the whole honeycomb that also takes every cell to another cell of the same type.

\section{Geodesic Regular Tree Structures}\label{sec:geodesic}

We use a similar definition of a regular tree structure (RTS) as in \cite{gentes}. As mentioned in the Introduction, we cannot label non-tree edges as Left/Right, so we need to use a different approach.

A regular tree structure for a honeycomb $(T,S,C,F)$ is $(Q, t:Q \ra T, e)$, where $e$ is a function from $\{(q,f): f\in F(t(q))\}$ to $Q \cup F^* \cup {P}$. $P$ denotes a parent connection; for every $q$,
there is at most one $f$ such that $e(q,f)=P$; and $F^*$ is the set of all finite sequences of faces. We call the case $e(q,f) \in Q$ a \emph{child connection}, and the case $e(q,f) \in F^*$ 
a \emph{side connection}. The sequence $e(q,f)\in F^*$ itself is called a \emph{side path}.

A regular tree structure lets us lazily generate a honeycomb $C$ as follows. Every cell $c\in C$ will have a state $q(c)$, and $t_c = t(q(c))$. 
We start with a tile $c_0$ in state $q_0$ which has no parent connections. Every new cell initially has no neighbors, i.e.,~$\nu(g,f)=NULL$ for every $f \in F(t(q(c)))$.
Cells are generated lazily. When querying for $\nu(c,f)$, we connect it to a tile as follows:

\begin{itemize}
\item If $e(q(c),f) = q'\in Q$, we create a new cell $c'=\nu(g,f)$ such that $q(c')=q'$, and set $\nu(c',f') = c$, where $(t(q),f) \sim (t(q'),f')$. In this case, $e(q',f')$ must equal $P$. This rule
arranges all tiles of $C$ in a tree.
\item If $e(q(c),f) = w \in F^*$, we generate a side connection by recursively constructing $c'=\nu(c',w)$, and connect $\nu(c,f)=c'$ and $\nu(c',f')=c$ (again, $(t(q),f) \sim (t(q'),f')$).
Here we extend $\nu: C \times F$ to $\nu: C \times F^*$ via $\nu(c,\epsilon)=c$ and $\nu(c,fw) = \nu(\nu(c,f),w)$ for $f \in F$, $w \in F^*$. We call this process the \emph{side connection procedure}.
In a correct RTS the side connection procedure will always terminate.
\item By the construction above, the situation $e(q(g),f) = q'\in Q$ can never occur.
\end{itemize}

By $\delta(c)$ we denote the shortest path from $c_0$ to $c$ in the adjacency graph. 
In a geodesic RTS, all the branches of the tree obtained are geodesic, that is, if $c$ is an $n$-th level descendant of $c_0$, then $\delta(c)=n$.

\section{Breaking the Symmetry of Regular Honeycombs}\label{sec:breaking}

Making a regular $\{p,q,r\}$ honeycomb fit the definition of fixed periodic honeycomb is not straightforward because that definition does not allow the tiles to have their own symmetries.
In this Section, we explain our method of breaking the intrinsic symmetry. This method is inspired by the method used for generating the manifold in David Madore's hyperbolic maze \cite{hypmaze}.
The two main ideas are as follows:
\begin{itemize}
\item Instead of $\bbX$, we use a closed manifold $M$ that can be exactly subdivided into $\{p,q,r\}$ cells. Every cell in such a manifold can then be given a fixed rotation. After taking the universal 
cover of such a manifold, every cell remains fixed, and we get tile types corresponding to the finitely many cells of $M$. 

\item In some cases there are simple examples of manifolds we could use:
for \{4,3,4\} we could use the torus obtained by gluing the opposite faces of the cube; for \{5,3,5\} we could use the Seifert-Weber space. In other cases, we can find an appropriate manifold by
performing the same computations as are used for generating the $\{p,q,r\}$ honeycomb in $\bbX^3$, but in a finite field instead of the group of isometries of $\bbX$.
\end{itemize}

Consider a cell $c_0$ in the regular $\{p,q,r\}$ honeycomb, let $f_0, f_1, f_2$ be its faces such that $f_1$ and $f_2$ are vertex-adjacent and both edge-adjacent to $f_0$, and let $v$ be one of the vertices of $f$. Let $c_1$ be the cell on the other side of $f_0$.
Let $e$ be the common edge of $f_0$ and $f_1$. Let $v$ be the common vertex of $f_0$, $f_1$ and $f_2$.
The group $G$ of orientation-preserving isometries of the regular $\{p,q,r\}$ tessellation is generated by the following three elements:
\begin{itemize}
\item $P$, which maps $c_0$ to $c_1$ and $c_1$ to $c_0$, and $e$ to $e$,
\item $R$, which rotates $c_0$ by mapping $f_1$ to $f_0$ to $f_2$,
\item $X$, which rotates $c_0$ to $c_0$ by mapping $f_0$ to $f_1$ and $f_1$ to $f_0$.
\end{itemize}

Note that the orders of $P$, $X$, $XP$, $R$ and $RX$ are 2, 2, $p$, $q$, and $r$, respectively. Also $XP=RPXR$.

Let $G'$ be a (finite) group generated by elements $P'$, $X'$, and $R'$. We say that $P'$, $X'$ and $R'$ are \emph{good} if 
the orders of elements $P'$, $X'$, $X'P'$, $R'$ and $R'X'$ are as above, i.e.,~2, 2, $p$, $q$, and $r$, and
$X'P'=R'P'X'R'$, and that the subgroup $H'$ generated by $X'$ and $R'$ is isomorphic to the group $H$ generated by $X$ and $R$, via an isomorphism $\phi$ which maps $X$ to $X'$ and $R$ to $R'$.
We obtain the set of types $T$ by taking one representative of every left coset. For every face $f$ of our cell, there is an isometry $I \in H$ which takes $f_0$ to $f$.
We set $\nu(t,f)=t'$ such that $t' = \phi(h) P' \phi(I)$, for some $h \in H$. The isometry $C_{t,f}$ is then $hPI$.

\begin{proposition}
The manifold $M_G$ obtained by taking a $\{p,q,r\}$ cell for every $t\in T$
and gluing them according to the $C_{t,f}$ isometries is a quotient space of $\bbX$, and thus yields a fixed periodic honeycomb.
\end{proposition}

The manifold $M_G$ obtained is highly symmetric, that is, every $g \in G$ corresponds to an isometry $I_g$ of $M_G$ (the two-dimensional version of this method for $\{7,3\}$ yields Hurwitz surfaces).
For our application, we do not need this high symmetry, but we would like the set $T$ to be as small as possible. This is achieved by taking a subgroup $K' \subseteq G'$ such that $K' \cap H'$
consists only of the neutral element. In $M_G$ we identify every point $x$ with $I_kx$ for every $k \in K$, thus obtaining a smaller manifold which is still obtained by gluing $\{p,q,r\}$ cells.

We find an appropriate group $G'$ for a hyperbolic honeycomb as follows. We represent points of $\bbH^3$ in the Minkowski hyperboloid model, as points $(x_0,x_1,x_2,x_3) \in \bbR^4$ such that
$x_3>0$ and $x_3^2-x_2^2-x_1^2-x_0^2=1$. The group $G$ can be seen as a subgroup of the group of Minkowski hyperboloid isometries, and therefore, the group of 4x4 invertible matrices over $\bbR$. 
The matrices of $X$ and $R$ also satisfy $X_{i,3}=R_{i,3}=X_{3,i}=R_{3,i}=0$ for $i=0,1,2$, and $X_{3,3}=R_{3,3}=1$. Since $P$, $X$, $R$ are isometries of Minkowski space,
we also have $M^{-1}=AM^TA$ for $M \in \{P,R,X\}$, where $A$ is the diagonal matrix with $(1,1,1,-1)$ on diagonal.

We generate $G'$ by considering a similar construction but over a finite field $\bbF_n$. For simplicity, we use the values of $n$ which are primes or squares of primes; in the case of squares of primes,
we consider only elements whose squares are multiples of 1. We find $P'$, $R'$, and $X'$,
4x4 matrices over $\bbF^n$ which satisfy the properties given in the last paragraph, and are also \emph{good}.

See Appendix \ref{sec:regular} for a list of manifolds (and thus fixed periodic honeycombs) obtained.

\section{The Main Algorithm}\label{sec:main}

Just like in \cite{gentes}, the core algorithm of construction GRTS is based on Angluin-style learning. We recall the main ideas.

\begin{itemize}
\item During the construction of GRTS we have no access to the honeycomb $C$ together with $\delta(c)$ for every cell, but rather, we need to lazily generate an approximation $G$.
Our implementation in the three-dimensional case instead is based on computing the isometries
$I_c$ numerically: we start with a root cell $c_0$ with some $I_{c_0}$, and when generating a connection $c'=\nu(c,f)$, we compute $I=I_{c_0} C_{t_{c_0},f}$ and if a cell $c''$ of appropriate type such that
$I_{c'}=I$ already exists, we know that $c'=c''$, otherwise we create a new cell. To reduce numerical errors, the isometries here are computed relative to a honeycomb whose structure is known;
we use the three-dimensional analog of the binary tiling \cite{rtvizfinal}. While numerical errors are a potential possibility in this approach, our experiments show that they do not occur for the honeycombs
we tested our methods on. An alternative implementation would not compute isometries $I_c$ and instead close loops if all the cells around an edge are known,
in the manner similar to the 2D case \cite{gentes}.
\item We denote by $\delta(c)$ the distance from $c$ to $c_0$; we have no access to $\delta(c)$, only to $\delta^+(c)$, the best currently known upper bound on $c$. We compute $\delta^+(c)$ as in the two-dimensional case \cite{gentes}. We also
use the shortcut method \cite{gentes} in order to avoid \emph{distance errors} ($\delta^+(c) \neq \delta(c)$) as much as possible.
\item We assume that the faces $F(t)$ are ordered for every $t \in T$. The parent of every cell $c \neq c_0$ is simply the cell $c'=\nu(c,f)$ such that $\delta(c')=\delta(c)-1$, and $f$ is the first face
with this property. Our algorithm will find the first face such that $\delta^+(c')=\delta^+(c)$, which may be a wrong face due to distance errors.
\item In \cite{gentes}, the core algorithm finds the state of $c$ by classifying the faces of $c'$ (which is $c$ or its descendant) into eight categories (child, parent, left/right edge of relative distance -1/0/+1). 
In the three-dimensional case, the categories are child, parent, and an extended side path $(p,d) \in F^* \times \bbZ^*$. Here $p$ is similar to the side path from the definition of GRTS; for every cell $c_i = \nu(c', p_1 \ldots p_i)$ we
also record $d_i = \delta^+(c_i)-\delta^+(c)$. The extended side path $(p,d)$ consists of a move from $c'$ to its parent, a \emph{middle sequence} of moves in the part of the honeycomb consisting of tiles $c''$ such that $\delta^+(c'')<
\delta^+(c')$, followed by the \emph{final moves}, that is at most two child edges. The middle sequence is found using breadth-first search in a deterministic way, i.e., if the relative distances in the neighborhood of $c_1$
and the relative distances in the neighborhood of $c_2$ are the same, the extended side path found will also be the same.
\item The cells are assigned to states based on the classification of all the faces of $c$ and possibly additional faces of descendants of $c$.
We get an example cell $s_q \in G$ for every state $q \in Q$. The RTS rule for face $f$ of state $q$
is obtained from the classification for the face $f$ of $s_q$. In case if $\nu(s_q,f)$ is a child, the rule is the state assigned to $\nu(s_q,f)$. (If the classifications are equal for $c_1$ and $c_2$ but
the states for children $\nu(c_1,f)$ and $\nu(c_2,f)$ are different, we improve the decision trees so that $c_1$ and $c_2$ are two different states \cite{gentes}.)
We actually get an \emph{extended} RTS, that is, one where for side connections we not only know the side path $p \in F^*$ but also the relative distances $d \in \bbZ^*$.
\end{itemize}

The way to verify the correctness of such a candidate extended GRTS is explained in the next section.

\subsection{Verifying the 3D GRTS}\label{sec:verify}

The verification of the correctness of a candidate extended GRTS is based on automata theory \cite{wpigroups}.

A correct GRTS will generate a tree; every cell $c \in C$ will have a path $w \in F^*$ leading to it. Let $L$ be the set of all such paths. From the construction we know that $L$ is a regular language.
For every $w \in L$, let $c(p)$ be the cell described by $w$, $q(w)\in Q$ be its state and $t(w)\in T$ be its type. 
For every $w \in L$ and every side face $f \in F(t(w))$, our extended GRTS gives a way to compute $w' \in L$ such that $c(w') = \nu(c(w), f)$. We will denote such $w'$ with $\nu(w, f)$. We also denote 
the length of $w$ by $\delta(w)$.

\begin{proposition}\label{prop:dist}
Suppose that for every $w \in L$ and for every side rule $e(q(w),f) = (p,d) \in F^* \times \bbZ^*$ the following property holds: for every $i$ from 1 to the length of $p$, 
we have $\delta(\nu(w,p_1\ldots p_i)) = \delta(w) + q_i$. Furthermore, the moves corresponding to the final moves are indeed child moves. Then the side connection procedure will always terminate.
\end{proposition}

\begin{proof}
Induction over $\delta(w)$.
\end{proof}

For every $t \in T$, let $\Gamma(t) \subseteq F^*$ be the set of all sequences of moves corresponding to cycles around every edge of $S_t$. Note that for every $c \in C$ and every $\gamma \in \Gamma(t_c)$
we have $\nu(c,\gamma) = c$.

\begin{proposition}\label{prop:cycle}
Suppose that $\bbX$ is simple connected, and for every $w \in L$ and for every $\gamma \in \Gamma(t(w))$ we have $\nu(w,\gamma)=w$. Then our GRTS generates every cell of $C$ at most once, that is, if $w_1 \neq w_2 \in L$, $c(w_1) \neq c(w_2)$.
\end{proposition}

\begin{proof}
To the contrary, suppose that $w_1 \neq w_2$ and $c(w_1) = c(w_2) = c$. In this case, there is a sequence of moves $u \in F^*$ such that $\nu(w_1, u) = w_2$. This $u \in F^*$ consists of the inverses of moves in $w_1$
followed by $w_2$. Geometrically, $u$ is a loop from $c$ to $c$, and since $\bbX$ is simple connected, it is homotopic to the empty loop. We can take a homotopy which never crosses the vertices and never crosses two edges at
the same time. Therefore, there is a sequence of paths $u_0, u_1, \ldots, u_n$ such that $\nu(c,u_i)=c$ for every $i$, $u_0$ is the empty path, $u_n=u$, and $u_i$ and $u_{i+1}$ always differ only in that they go around some edge
in two different ways, that is, the difference between $u_i$ and $u_{i+1}$ is one of the cycles in $\Gamma(t)$ for some $t$. By induction, we have that $\nu(w_1,u_i)=w_1$ for every $i$, and thus, $\nu(w_1,u)=w_1$, hence $w_1=w_2$.
\end{proof}

Therefore, to guarantee that our GRTS generates every cell of $C$ exactly once and that it terminates, we need to verify that the assumptions of both Propositions \ref{prop:dist} and \ref{prop:cycle} hold.

\def\blank{\square}

For $w, u \in F^*$, let ${w \choose u} \in (F \cup \{\blank\})^2$ be the word of length $\max(|w|, |u|)$ such that ${w \choose u}_i = (w_i, u_i)$ if $i \leq |w|$ and $i \leq |u|$; if $i > |w|$ or $i > |u|$, we take $\blank$
instead of $w_i$ or $u_i$ respectively. For $t \in T$ and $f \in F(t)$, the set $X_{t,f} = \{{w \choose u}: w \in L, t(w)=t, u=\nu(w,f)\}$ is a regular language. This language is recognized by a deterministic transducer $A_{t,f}$.
A {\emph transducer} is an automaton over $(F \cup \{\blank\})^2$ which recognizes a relation $R \subseteq F^* \times F^*$: $(w,u) \in R$ iff ${w \choose u}$ is accepted by the transducer.
In this paper, all the transducers are expected to recognize functions. Standard automata-theoretic
techniques can be used to create a transducer representing identity function on a regular language $L \subseteq F^*$, compose two transducers, find a $w \in L$ which is not in the domain of a transducer, or find a $w \in L$ such
that ${w \choose u_1}$ and ${w \choose u_2}$ where $u_1 \neq u_2$ are both accepted by a given transducer (i.e., the transducer does not define a function).

The following observation is straightforward:

\begin{proposition}\label{prop:mn}
Let $w,u,w',u' \in F^*$ be of the same length.
We say that ${w \choose u}$ and ${w' \choose u'}$ are equivalent if 
$q(w)=q(w')$, $q(u)=q(u')$, and such that there exists an isometry $J$ such that $I_{c(w)}J = I_{c(u)}$ and $I_{c(w')}J=I_{c(u)}J$.

If ${w \choose u}$ and ${w' \choose u'}$ are equivalent, they are also Myhill-Nerode equivalent for every $X_{t,f}$, that is, for every $w'',u'' \in F^*$, 
we have ${ww'' \choose uu''} \in X_{t,f}$ if and only if ${w'w'' \choose u'u''} \in X_{t,f}$.
\end{proposition}

By Myhill-Nerode theorem, the states of $A_{t,f}$ correspond to the classes of Myhill-Nerode equivalence above. Therefore, we can construct $A_{t,f}$ iteratively as follows.
We start with an empty transducer, and iteratively add transitions to it. Find a word $w \in L$ such that $t(w)=t$ and such that there is no $u$ such that ${w \choose u}$ is in the language
of the current $A_{t,f}$. We compute $u = \nu(w,f)$ and add new transitions and accepting states to $A_{t,f}$ so that ${w \choose u}$ will be accepted by it.
Each state of $A_{t,f}$ corresponds to one class of equivalence from Proposition \ref{prop:mn} (that is, if ${w \choose u}$ and ${w' \choose u'}$ are equivalent,
we keep the invariant that $A_{t,f}$ reaches the same state after reading both words).

During the computation of $\nu(w,f)$ we verify whether Proposition \ref{prop:dist} holds for it; if no, we can find a cell in $G$ for which our candidate GRTS fails. We add such cell to
the set $S$ of sample cells, and go back to the Angluin-style algorithm. A similar thing is done if we detect that $A_{t,f}$ accepts both ${w \choose u_1}$ and $w \choose u_2$.

Once the tranducers $A_{t,f}$ are computed for every $t \in T$ and $f \in F(t)$, it is straightforward to verify that Properties \ref{prop:dist} and \ref{prop:cycle} hold for every
$w \in L$. For example, to verify whether Proposition \ref{prop:cycle} works for $\gamma \in \Gamma(t(w))$, we construct a transducer $I_t$ which recognizes the language $\{{w \choose w}: w \in L, t(w)=t\}$.
Then we compose this transducer with $A_{t,\gamma_1}$, $A_{\nu(t,\gamma_1),\gamma_2}$, \ldots, $A_{\nu(t,\gamma_1\gamma_2\ldots\gamma_{n-1}), \gamma_n}$. The transducer $I'_t$ obtained after all these compositions
should be again $I_t$. If not, we can again find $w \in L$ for which ${w \choose w}$ is not in the language of $I_t'$, find a responsible cell in $G$, and add that cell to the set $S$ of sample cells, so the
next iteration of the Angluin-style algorithm will find a new candidate GRTS which avoids the problem detected.

\subsection{Other Details}\label{sec:details}

Here we list the minor details of the implementation we have used for our experiments.
\begin{itemize}
\item Like in \cite{gentes} we actually build our trees not from a single root, but from a root for every possible tile type $t \in T$.
\item As explained above, for every cell we want to classify, we need to build extended side paths for all the side connections. It is possible that such extended side paths could be built incorrectly because
of some cells or their connections not being known, or distance errors. To prevent this, we use the following approach. Suppose we find out that we can reach $c'=\nu(c,f)$ for some $c\in G$ and $f\in F(t_c)$
by following a path $c=c_0, c_1, \ldots, c_k=c'$. Then we remember the subtrees connecting $c$ to every $c_i$, let $c'_i$ be the least common ancestor of $c$ and $c_i$. 
Later, when determining side connections for some other $d \in G$, we check already constructed subtrees and see whether for any of them, the path from $c$ to $c'_i$ matches tha path from $d$ to one
of its ancestors $d'_i$ (that is, the same sequence of tile types and their faces). If so, we follow the path from $c'_i$ to $c_i$ but starting from $d'_i$, and thus construct $d_i$, which we can expect to
be on the side path we want for $d$.
\item Before starting the transducer-based algorithm described above, we perform the following optimization. Our candidate tree structure lets us compute, assuming that the candidate tree structure is correct,
for every state, the representative of that state which is the closest to the root (in case of a tie, the one with the first root, and lexicographically first path from the root, according to some order on
tile types and faces). We check if the states of such candidates are as expected. We also check if there are any states that our classifier has found in $G$ but turn out to not be reachable in the candidate
tree; if so, we compute the states for all the cells from the cells which gave us such states to their roots -- we are guaranteed to find some inconsistency with our candidate tree structure. If no
inconsistencies are found, the next iteration of our algorithm will take the set of closest representatives as the set $S$. We have found that such an approach to detect candidate tree errors and 
optimizations significantly improves the performance of our algorithm.
\item To verify \ref{prop:dist} (and similarly \ref{prop:cycle}) we essentially need to compute $\nu(w,f)$ for every $w$ and every $f \in F(t(w))$ and see if the distances are correct.
To compute $\nu(w,f)$ we typically do not need to know whole $w$, but rather only a suffix of it. A quick method of preliminary verification whether Propositions \ref{prop:dist} and \ref{prop:cycle} hold
is to verify them for all suffixes of length at most $l$. In our implementation, before starting the transducer-based algorithm, we do such a verification for $l \leq 3$.
\item While computing $\nu(w,f)$ and verifying Proposition \ref{prop:dist} for it, our implementation only checks the distances -- according to Proposition \ref{prop:dist}, we should also check whether
the final moves indeed correspond to child connections. We also skip the final transducer-based check of Proposition \ref{prop:dist} for all words. This reduces the time needed for our experiments, while
not fully verifying that our side connection procedure from Section \ref{sec:geodesic} terminates. We do not consider this a big problem: we have checked some words in the process of building the transducers,
and also, since the transducers have been built successfully, we know that we could build the side connections based on the transducers rather than using the side connection procedure.
This does not affect the correctness of coordination sequences.
\end{itemize}

\section{Conclusion}\label{sec:concl}

The algorithm described above correctly generates a GRTS for the fixed periodic honeycombs we have tested it on. See Appendix \ref{sec:tested} for the list of such honeycombs with their characteristics, and 
\url{https://figshare.com/articles/software/Generating_tree_structures_for_hyperbolic_honeycombs/20713306} for the implementation and GRTS obtained.

\section*{Acknowledgments}
This work has been supported by the National Science Centre, Poland, grant UMO-2019/35/B/ST6/04456.

\bibliographystyle{alpha}
\def\ext#1{#1}
\bibliography{../master}

\appendix

\section{Regular Honeycombs}\label{sec:regular}

This table lists the manifolds (and thus fixed periodic honeycombs) obtained using the methods in Section \ref{sec:breaking}.
The left column ($pqr$) lists the honeycomb (${p,q,r}$). The field column gives the size of the field used. The cells column is the number of cells in the
symmetric manifold. The quotients column lists the humbers of cells in the available quotient spaces. The hash column is used to identify a particular symmetric manifold in the RogueViz engine.
The * symbol denotes the fixed periodic honeycomb we have chosen. For \{4,3,4\} we use the torus obtained by gluing the opposite faces of the cube, and for \{5,3,5\} we use the Seifert-Weber space. 

\begin{center}
\begin{longtable}{llrrrl}

\hline
$pqr$ & hash & prime & field & cells & quotients \\
\hline

336 & D29C2418 & 3 & 3 & 10 & 5,1* \\
336 & 7769A558 & 5 & 25 & 650 & 50 \\
336 & 640FB3D4 & 7 & 7 & 28 & 14,1 \\
336 & C734F868 & 7 & 7 & 28 & 14,1 \\
336 & E3F6B7BC & 7 & 49 & 672 & 336,28,16 \\
336 & 885F1184 & 7 & 49 & 672 & 336,42,28,16 \\

\hline

344 & B23AF1F4 & 3 & 3 & 5 & 1* \\
344 & 4F9920E0 & 3 & 3 & 5 & 1 \\
344 & 6DBBAAA0 & 3 & 3 & 10 & 5,1 \\
344 & F81E97B0 & 3 & 3 & 10 & 5,1 \\
344 & F790CEA4 & 3 & 3 & 30 & 6,3 \\
344 & C95EC8B8 & 3 & 3 & 30 & 3 \\
344 & 16518434 & 3 & 9 & 16 & 8,4,2 \\
344 & 558C8ED0 & 5 & 5 & 600 & 20,25 \\
344 & 1EC39944 & 5 & 5 & 600 & 20,25 \\
344 & AF042EA8 & 5 & 25 & 2400 & 1200,600,60,40 \\
344 & EC29DCEC & 5 & 25 & 2600 & 1300,650,100 \\
344 & D26948E0 & 5 & 25 & 2600 & 1300,650,100 \\

\hline

345 & F978E264 & 3 & 3 & 30 & 6,3* \\
345 & 02ADCAA4 & 3 & 3 & 30 & 6,3 \\
345 & 7EFE8D98 & 5 & 25 & 650 & 50,25 \\
345 & F447F75C & 11 & 11 & 55 & 5 \\
345 & 58A698B8 & 19 & 19 & 285 & 15 \\
345 & 6FA03030 & 19 & 19 & 285 & 57,15 \\

\hline

353 & 1566EBAC & 5 & 25 & 130 & 10 \\
353 & 5A2E2B88 & 11 & 11 & 11 & 1* \\

\hline

354 & 58A8E850 & 5 & 5 & 2 & 1* \\
354 & 363D8DA4 & 11 & 11 & 22 & 11,2 \\
354 & 9CD5E744 & 11 & 11 & 22 & 11,2 \\
354 & F04BA28C & 19 & 19 & 114 & 57,6 \\

\hline

355 & AF448B14 & 5 & 5 & 1* \\
355 & F42F2904 & 5 & 5 & 1 \\
355 & 47F0C740 & 5 & 5 & 120 & 6,4 \\
355 & 7BAFB45C & 11 & 11 & 11 & 1 \\
355 & 6453A3FC & 11 & 11 & 11 & 1 \\

\hline

435 & EB201050 & 5 & 25 & 650 & 25 \\
435 & 65CE0C00 & 11 & 11 & 55 & 11,5* \\
435 & 5641E95C & 11 & 11 & 55 & 11,5 \\

\hline

436 & 235F7508 & 2 & 4 & 2 & 1* \\
436 & C02F2A80 & 2 & 4 & 2 & 1 \\
436 & DFC6B8C0 & 2 & 4 & 8 & 4,2 \\
436 & 4D3C8B14 & 3 & 9 & 8 & 1 \\
436 & FF82A214 & 5 & 25 & 650 & 50,25 \\
436 & 4546E270 & 5 & 25 & 650 & 25 \\
436 & C4884090 & 7 & 7 & 28 & 14,4,2 \\
436 & 5230B364 & 7 & 7 & 28 & 14,2 \\
436 & 2D051038 & 7 & 7 & 14 & 7,2 \\
436 & F0997060 & 7 & 7 & 14 & 7,2 \\
436 & 1D1227CC & 7 & 49 & 672 & 336,112,42,14,12 \\
436 & B2B4B3D4 & 7 & 49 & 672 & 336,14,12 \\
436 & 6C29B2A4 & 13 & 13 & 91 & 13,7 \\
436 & 06F4054C & 13 & 13 & 91 & 13,7 \\
436 & DE4912E0 & 13 & 13 & 182 & 91,13,7 \\
436 & 417466F0 & 13 & 13 & 182 & 91,13,7 \\

\hline

534 & 0C62E214 & 5 & 25 & 130 & 10 \\
534 & 72414D0C & 5 & 25 & 260 & 10 \\
534 & 831E2D74 & 11 & 11 & 22 & 2 \\
534 & 5FC4CFF0 & 11 & 11 & 22 & 2* \\

\hline

535 & DCC3CACE & 5 & 5 & 1 \\
535 & F78E1C56 & 5 & 5 & 1 \\
535 & 9EF7A9C4 & 5 & 5 & 120 & 10,8,6,4 \\
535 & 5254DA16 & 19 & 19 & 57 & 3 \\
535 & A5C8752E & 19 & 19 & 57 & 3 \\

\hline

536 & BB5AEE10 & 5 & 5 & 120 & 8,6,4 \\
536 & 61385498 & 5 & 5 & 2 & 1* \\
536 & B009EB44 & 5 & 5 & 2 & 1 \\
536 & 3BA5C5A4 & 5 & 25 & 130 & 10 \\
536 & 9FDE7B38 & 5 & 25 & 260 & 130,10 \\
536 & 885F1184 & 7 & 49 & 672 & 336,28,16 \\
\hline
\end{longtable}
\end{center}

\section{Tested Honeycombs}\label{sec:tested}

In this Appendix we list the honeycombs we have tested our algorithm on.
For each honeycomb we give the number of tile types, the number of states, and the obtained coordination sequence
(i.e., the sequene $(c_k)$ where $c_k$ cnumber of cells in distance $k$ from some central tile).

\begin{table}
\begin{tabular}{|lll|p{25em}|}
\hline
$pqr$ & tiles & states & coordination sequence \\
\hline
$\{3,3,6\}$ & 1 & 35 & 1, 4, 12, 30, 72, 168, 390, 900, 2076, 4782, 11016, 25368, 58422, 134532, 309804, 713406, 1642824\ldots \\

$\{3,4,4\}$ & 1 & 50 & 1, 8, 44, 224, 1124, 5624, 28124, 140624, 703124, 3515624, 17578124, 87890624\ldots \\

$\{3,4,5\}$ & 3 & 207 & 1, 8, 56, 368, 2408, 15776, 103328, 676736, 4432304, 29029472, 190128992\ldots \\

$\{3,5,3\}$ & 1 & 569 & 1, 20, 260, 3212, 39470, 484760, 5953532, 73117640, 897985850, 11028509072, 135445355180\ldots \\

$\{3,5,4\}$ & 1 & 83 & 1, 20, 350, 6080, 105590, 1833740, 31845830, 553053800, 9604664270, 166800365060, 2896755264110\ldots \\

$\{3,5,5\}$ & 1 & 171 & 1, 20, 380, 7160, 134900, 2541680, 47888240, 902270720, 16999841000, 320296987760, 6034771758320\ldots \\

$\{4,3,4\}$ & 1 & 52 & 1, 6, 18, 38, 66, 102, 146, 198, 258, 326, 402, 486, 578, 678, 786, 902, 1026, 1158, 1298, 1446\ldots \\

$\{4,3,5\}$ & 5 & 3175 & 1, 6, 30, 126, 498, 1982, 7854, 31014, 122562, 484422, 1914254, 7564542, 29893554\ldots \\

$\{4,3,6\}$ & 1 & 67 & 1, 6, 30, 138, 630, 2862, 13002, 59046, 268158, 1217802, 5530518, 25116174, 114062154, 517999686\ldots \\

$\{5,3,4\}$ & 2 & 118 & 1, 12, 102, 812, 6402, 50412, 396902, 3124812, 24601602, 193688012, 1524902502, 12005532012\ldots \\

$\{5,3,5\}$ & 1 & 1583 & 1, 12, 132, 1392, 14592, 153092, 1605972, 16846332, 176716302, 1853736312, 19445500172\ldots \\

$\{5,3,6\}$ & 1 & 163 & 1, 12, 132, 1422, 15312, 164832, 1774422, 19101612, 205628532, 2213587422, 23829228912, 256521221232, 2761446340422\ldots \\
\hline
\end{tabular}
\caption{Characteristics of regular honeycombs\label{tab:regular}}
\end{table}

\begin{table}
\begin{tabular}{|lll|p{25em}|}
\hline
honeycomb & tiles & states & coordination sequence \\
\hline
$\{3,3,6\}$ c0 & 2 & 122 & 1, 6, 18, 51, 120, 282, 660, 1512, 3516, 8052, 18636, 42780, 98748, 227052, 523404, 1204476, 2774892, 6388140\ldots \\
$\{3,3,6\}$ c7 & 2 & 141 & 1, 4, 9, 16, 25, 37, 53, 74, 101, 135, 179, 237, 313, 411, 537, 700, 912, 1188, 1546, 2009, 2608, 3385, 4394, 5703, 7399, 9596, 12444, 16138, 20929, 27140\ldots \\
$\{3,4,4\}$ c0 & 4 & 128 & 1, 6, 21, 62, 174, 480, 1316, 3600, 9840, 26888, 73464, 200712, 548360, 1498152, 4093032, 11182376\ldots \\
$\{3,4,4\}$ c7 & 2 & 131 & 1, 4, 9, 17, 29, 46, 70, 104, 152, 219, 314, 449, 639, 907, 1286, 1821, 2576, 3643, 5150, 7277, 10281, 14524, 20515, 28975, 40923\ldots \\
$\{3,5,3\}$ c0 & 10 & 5144 & 1, 6, 24, 72, 210, 581, 1581, 4308, 11664, 31634, 85740, 232302, 629526, 1705766, 4622202, 12524754\ldots \\
$\{3,5,3\}$ c7 & 2 & 465 & 1, 4, 9, 16, 26, 41, 62, 90, 128, 181, 254, 352, 483, 660, 900, 1224, 1661, 2252, 3052, 4133, 5592, 7562, 10224, 13821, 18680, 25244, 34113\ldots \\
$\{4,3,4\}$ c0 & 3 & 535 & 1, 8, 30, 68, 126, 180, 286, 348, 510, 572, 798, 852, 1150, 1188, 1566, 1580, 2046, 2028, 2590, 2532, 3198\ldots \\
$\{4,3,4\}$ c7 & 48 & 2164 & 1, 4, 9, 17, 28, 42, 60, 81, 105, 132, 162, 196, 233, 273, 316, 362, 412, 465, 521, 580, 642, 708, 777, 849, 924, 1002, 1084, 1169, 1257, 1348, 1442, 1540, 1641\ldots \\
$\{4,3,5\}$ bs1 & 7 & 2420 & 1, 32, 302, 2732, 24302, 216032, 1920002, 17064032, 151656302, 1347842732, 11978928302, 106462512032\ldots \\
$\{4,3,5\}$ bs2 & 74 & 18279 & 1, 32, 122, 452, 1202, 4232, 10922, 37832, 97202, 336452, 864122, 2990432, 7680002, 26577632, 68256122, 236208452, 606625202, 2099298632\ldots \\
$\{4,3,5\}$ d2 & 34 & 2094 & 1, 12, 42, 152, 402, 1272, 3242, 10092, 25602, 79532, 201642, 626232, 1587602, 4930392, 12499242, 38816972\ldots \\
$\{4,3,5\}$ s2 & 40 & 7793 & 1, 6, 21, 58, 141, 322, 720, 1602, 3551, 7842, 17283, 38074, 83898, 184894, 407415, 897630, 1977647, 4357242, 9600216\ldots \\
$\{4,3,6\}$ s2 & 200 & 11843 & 1, 6, 21, 61, 165, 435, 1137, 2961, 7701, 20019, 52029, 135213, 351381, 913131, 2372937, 6166497, 16024725\ldots \\
$\{5,3,4\}$ c0 & 12 & 4116 & 1, 10, 55, 237, 970, 3900, 15622, 62595, 250605, 1003462, 4017650, 16086450, 64408017, 257882475, 1032531650, 4134137032\ldots \\
$\{5,3,4\}$ c7 & 2 & 336 & 1, 4, 9, 17, 29, 46, 70, 103, 148, 210, 295, 411, 569, 783, 1074, 1470, 2008, 2740, 3736, 5091, 6934, 9440, 12848, 17483, 23786, 32358\ldots \\
\hline
\end{tabular}
\caption{Characteristics of subdivided honeycombs\label{tab:irregular}}
\end{table}

\begin{figure*}[h]
\begin{center}
  \includegraphics[width=.4\linewidth]{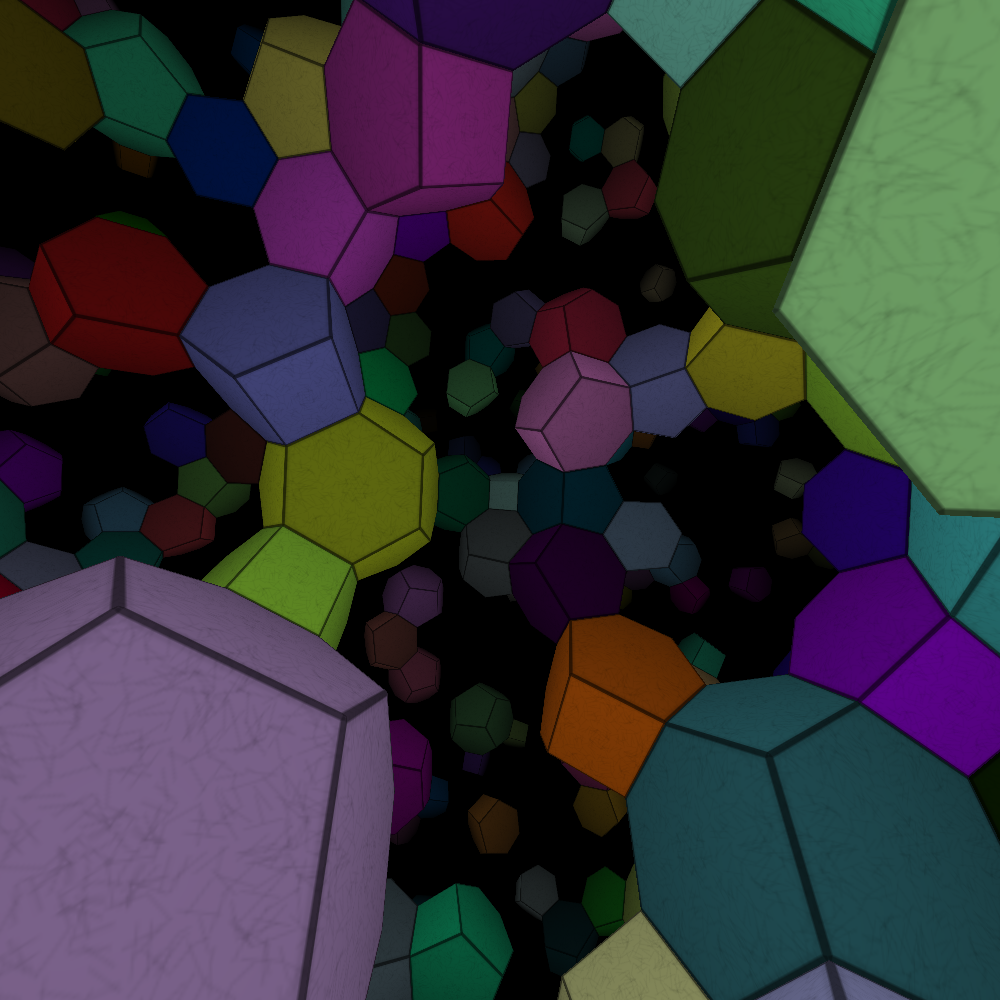}
  \includegraphics[width=.4\linewidth]{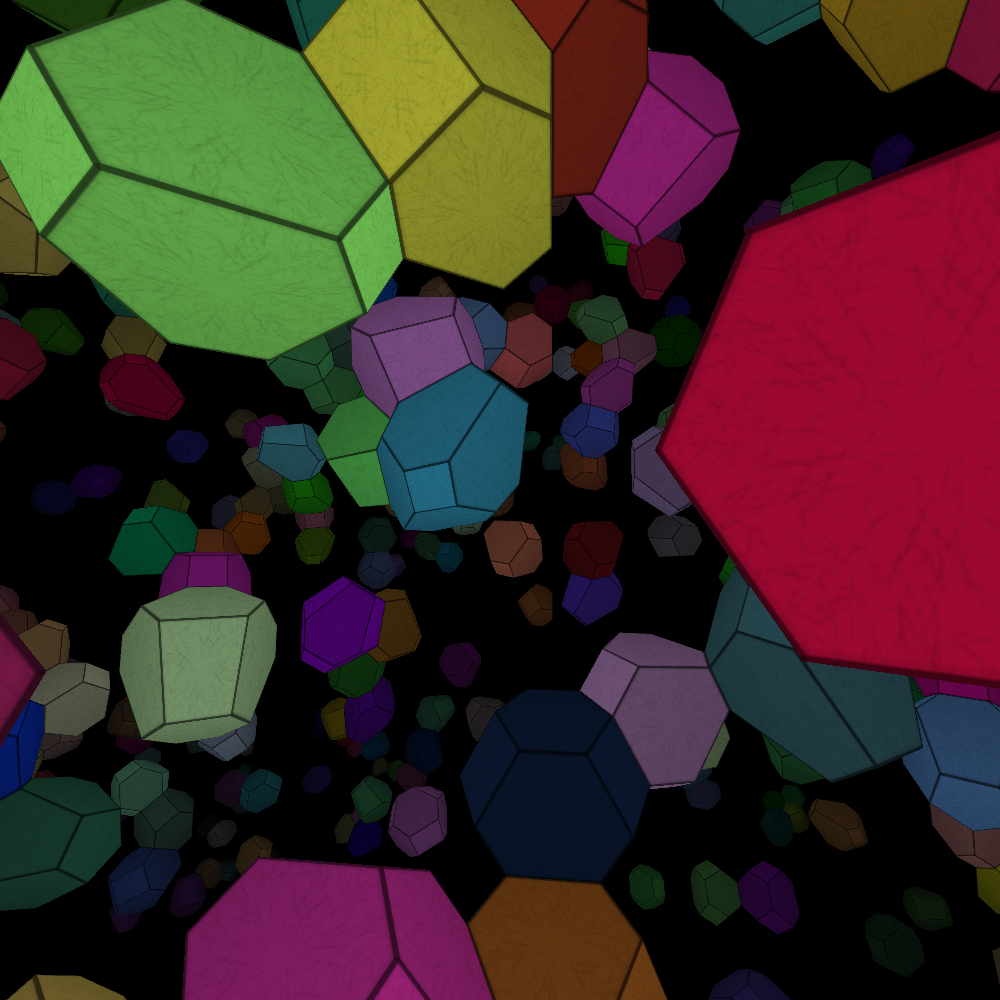}
  \includegraphics[width=.4\linewidth]{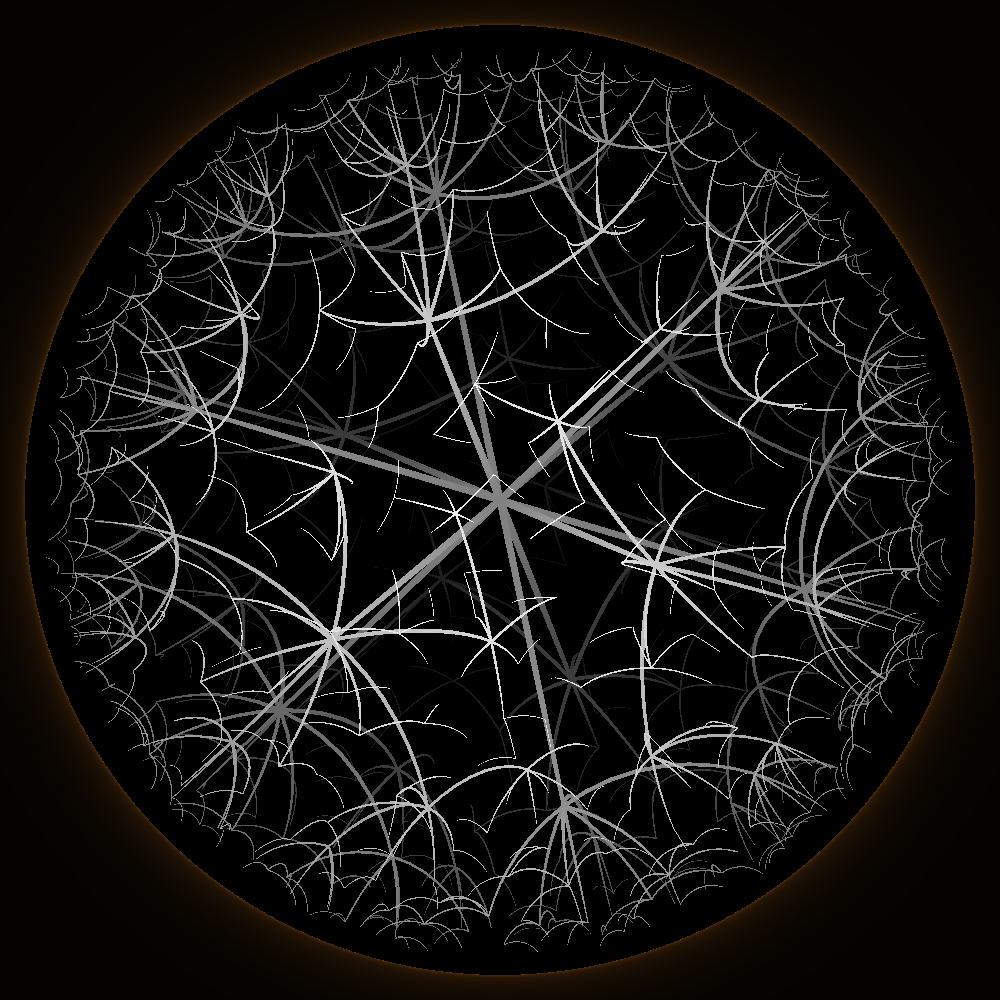}
  \includegraphics[width=.4\linewidth]{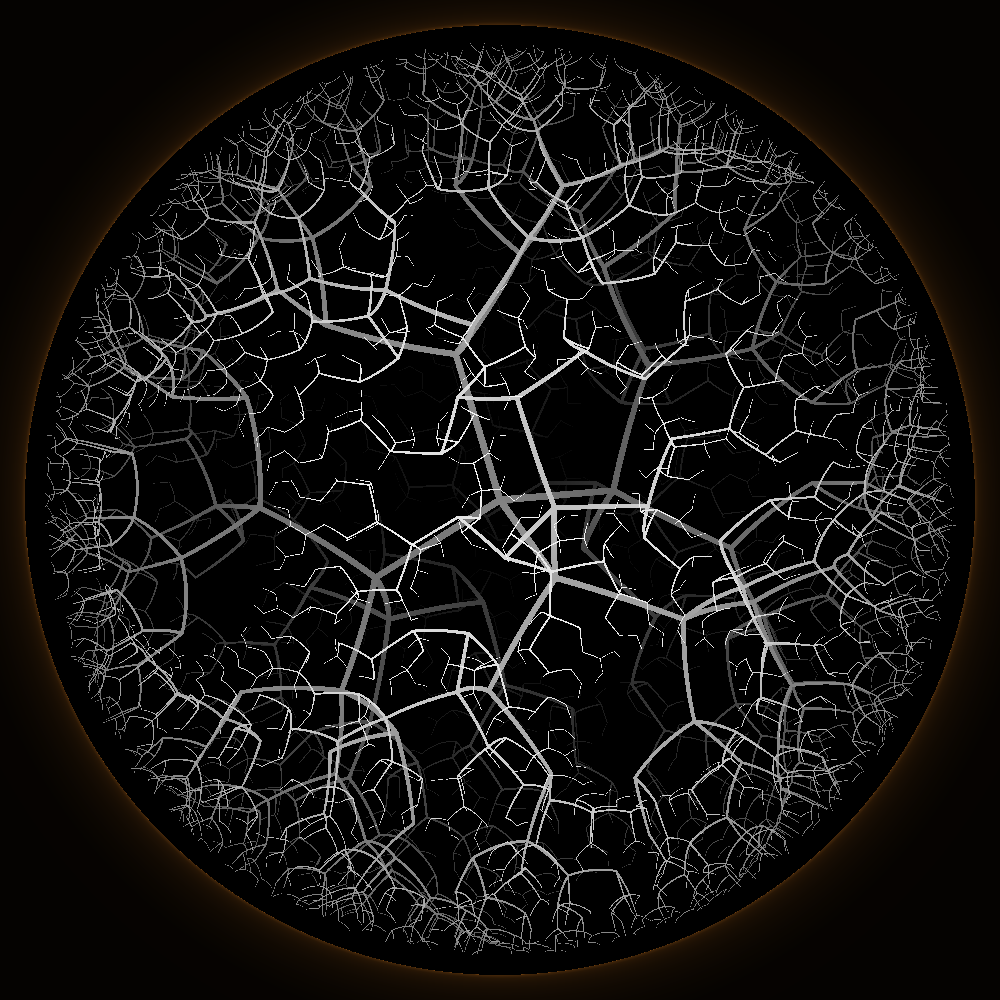}
  \includegraphics[width=.4\linewidth]{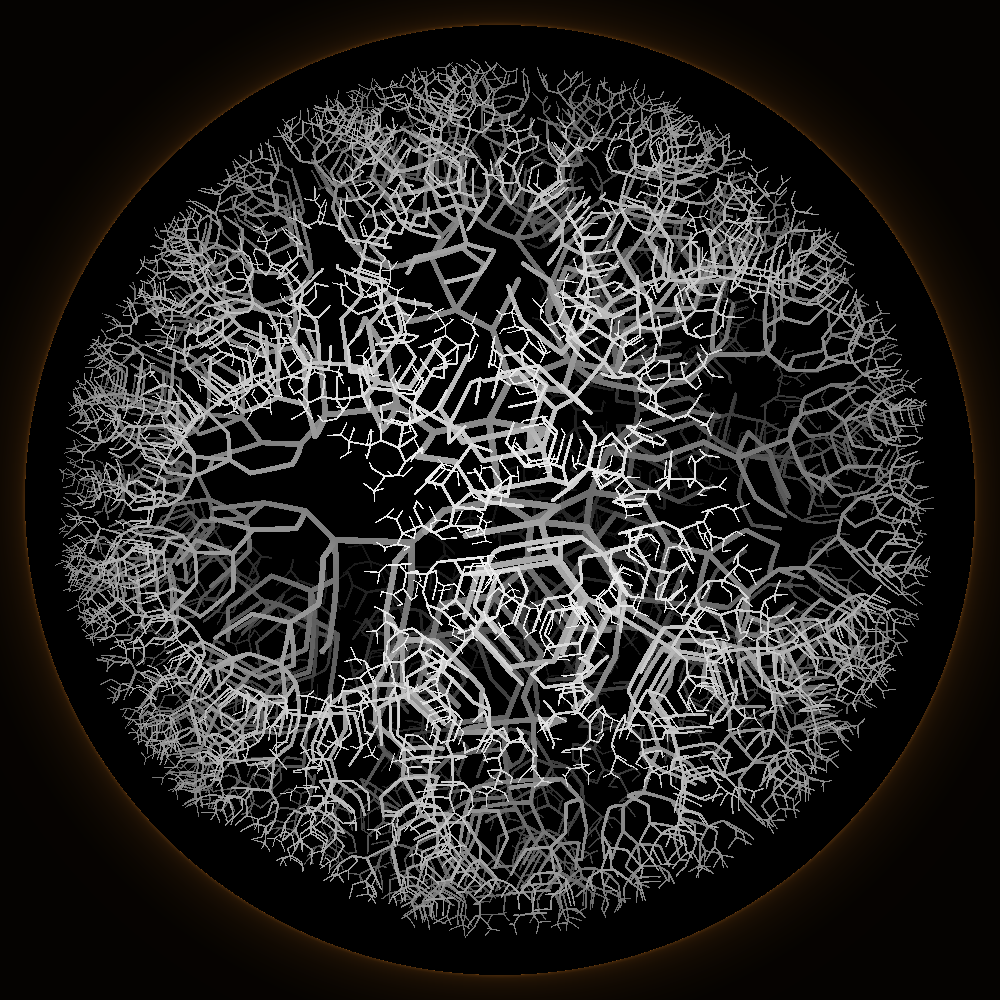}
  \includegraphics[width=.4\linewidth]{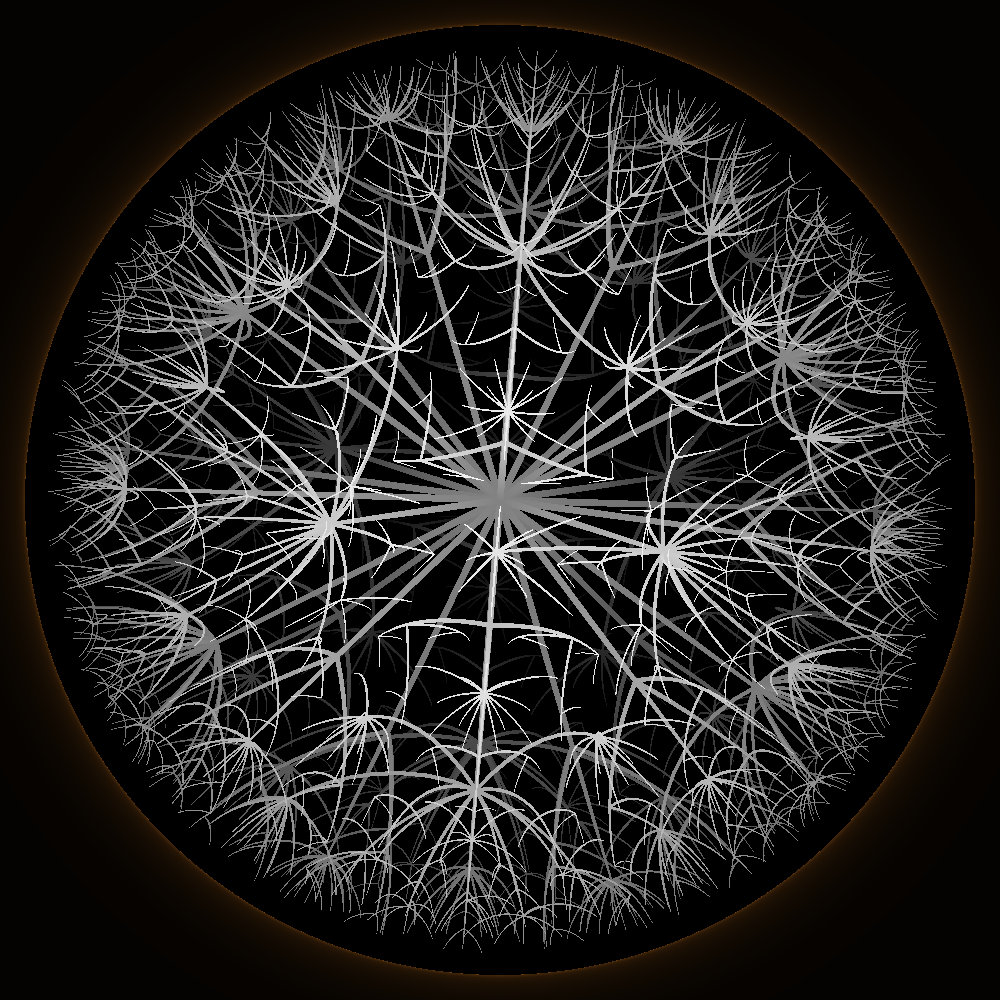}
\end{center}
\caption{Top row: the honeycombs \{4,3,5\}-bs1 and \{4,3,5\}-bs2, with some cells randomly filled. \label{hcimages}
Bottom row: The tree structures for: \{5,3,4\}, \{3,3,6\}, \{4,3,4\}-c0, \{3,4,4\}-c7, \{4,3,5\}-bs1, \{4,3,5\}-bs2. Poincar\'e ball model. \label{treeball}}
\end{figure*}

\begin{figure*}[h]
\begin{center}
  \includegraphics[width=.4\linewidth]{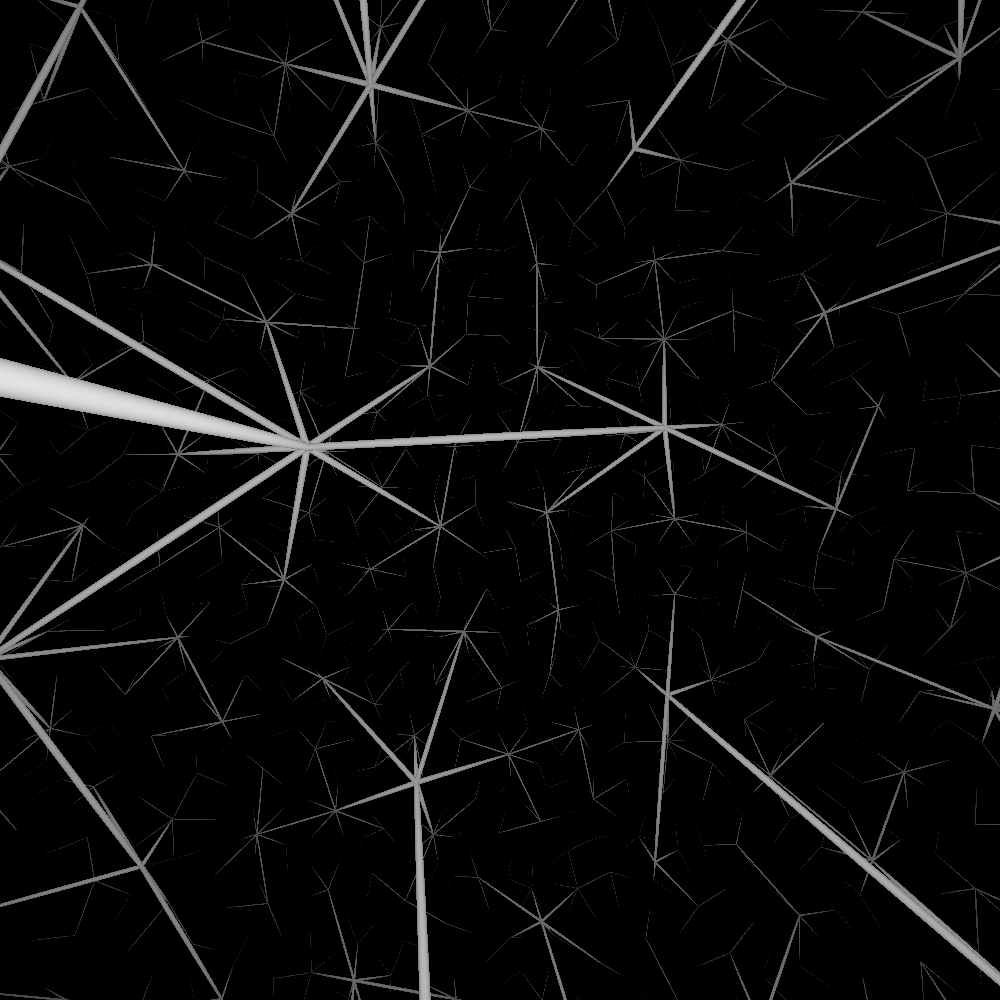}
  \includegraphics[width=.4\linewidth]{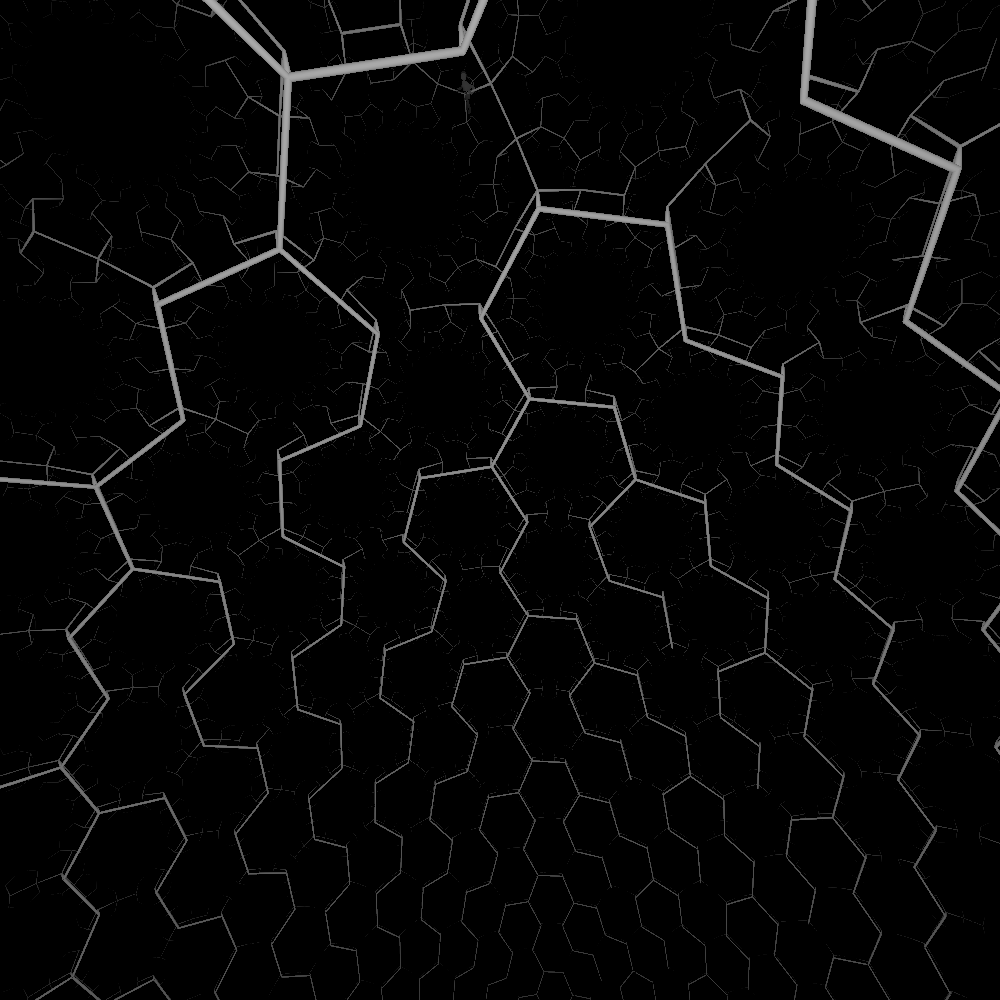}
  \includegraphics[width=.4\linewidth]{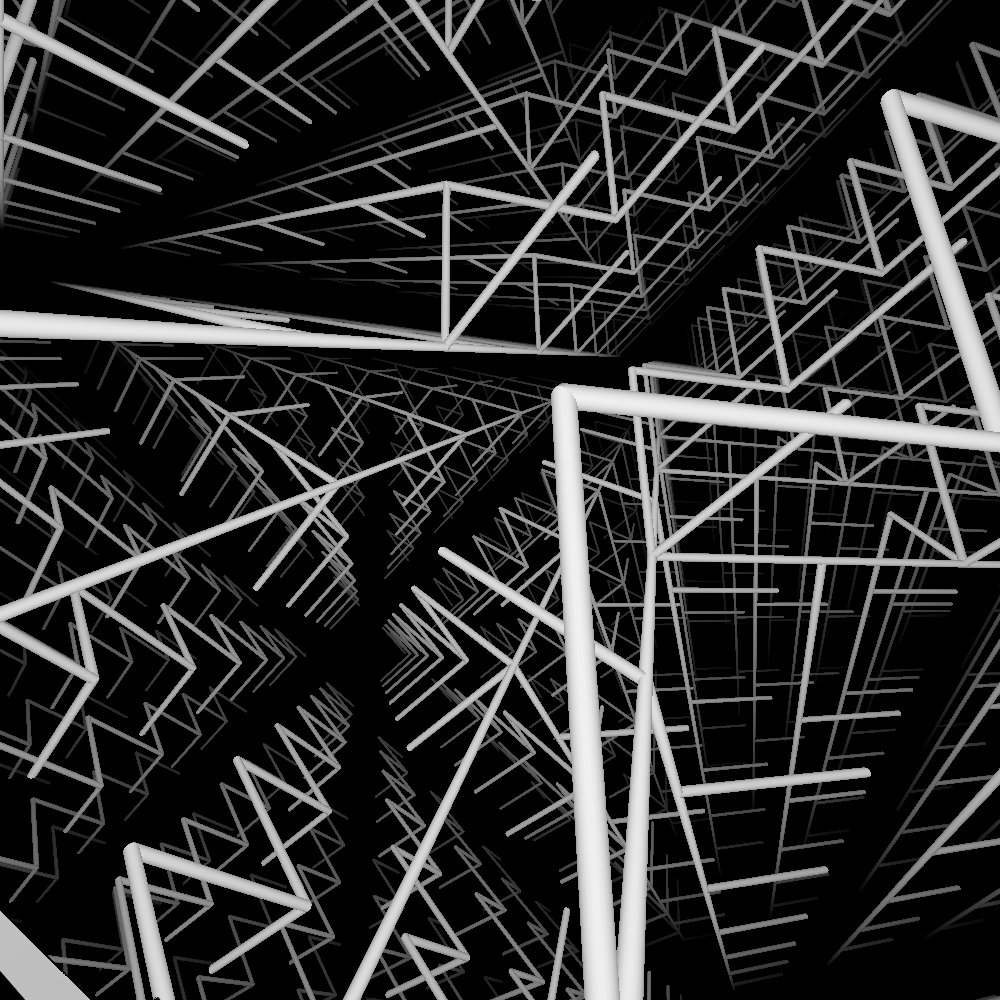}
  \includegraphics[width=.4\linewidth]{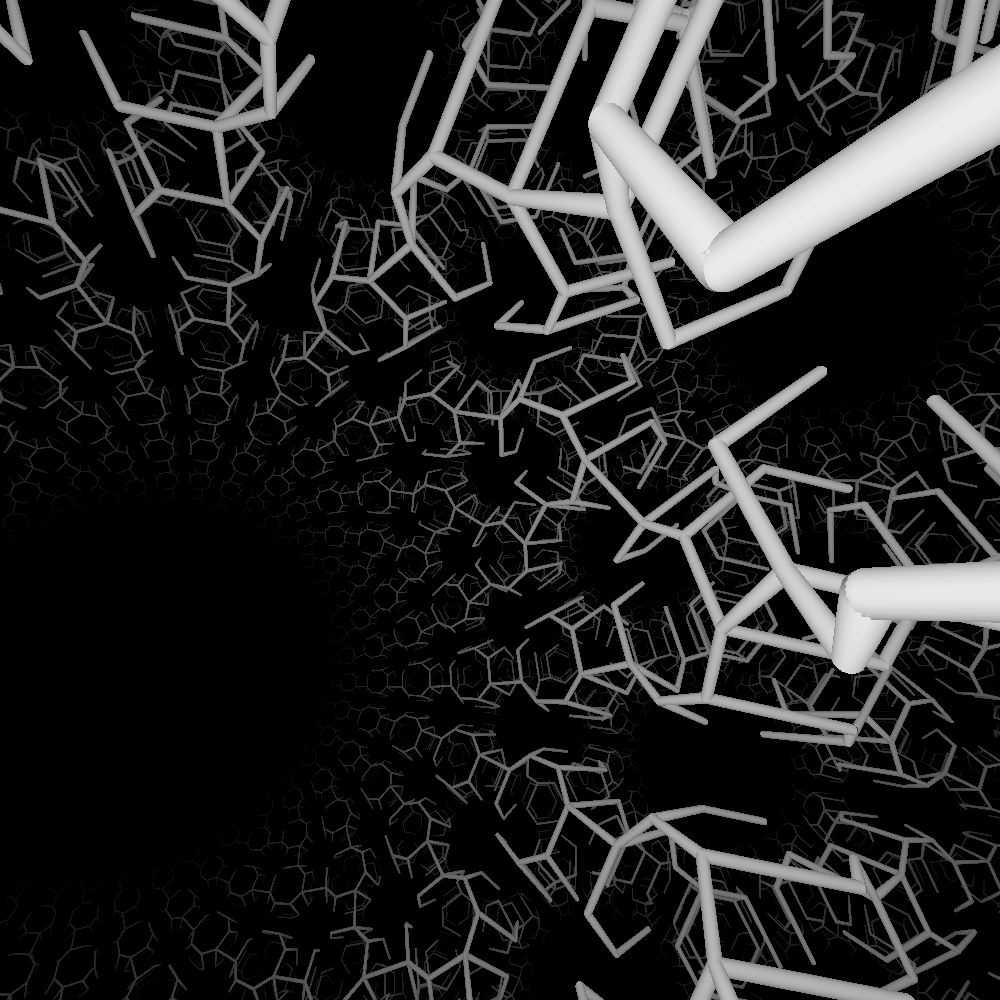}
  \includegraphics[width=.4\linewidth]{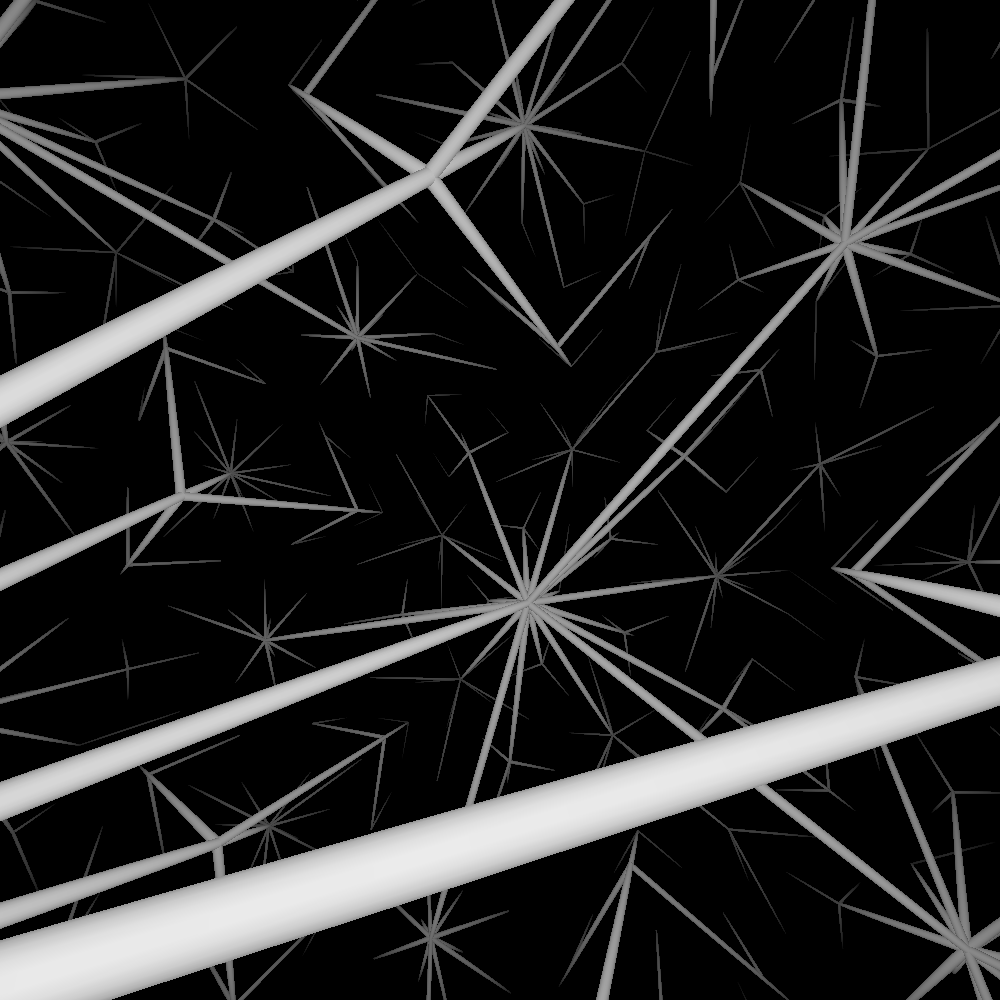}
  \includegraphics[width=.4\linewidth]{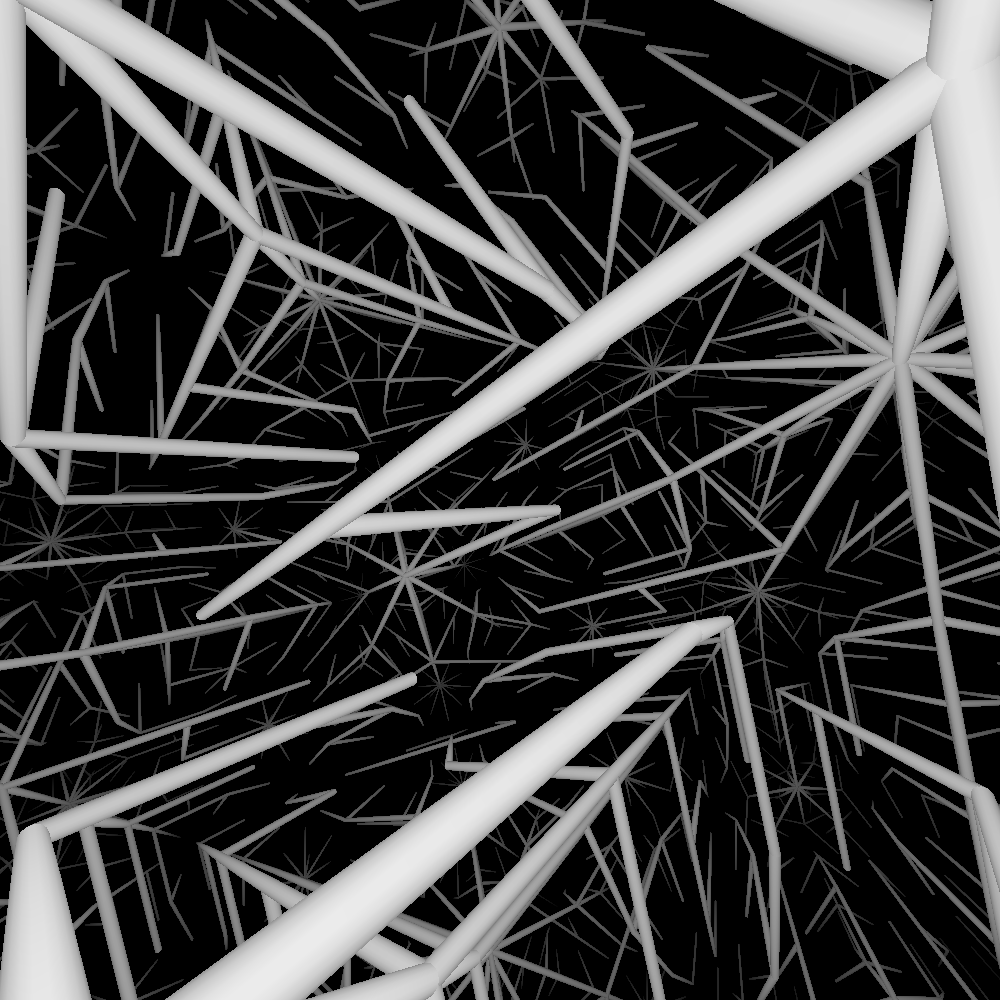}
\end{center}
\caption{The tree structures for: \{5,3,4\}, \{3,3,6\}, \{4,3,4\}-c0, \{3,4,4\}-c7, \{4,3,5\}-bs1, \{4,3,5\}-bs2. Native perspective projection.\label{treeimages}}
\end{figure*}

See Figures \ref{treeball} and \ref{treeimages} for example images of trees obtained.
In Table \ref{tab:regular} we provide descriptions of regular honeycombs. In table \ref{tab:irregular} we similarly describe subdivided regular honeycombs. We consider the following
regular subdivisions: (we denote the number of faces of \{p, q\} by $n$)
\begin{itemize}
\item c7: Coxeter tetrahedra (each cell subdivided into $2np$)
\item c0: each cell subdivided into $n$ pyramids,  then the pyramids are glued into double pyramids
\item s2: each cube subdivided into 8 cells by three orthogonal symmetry planes
\item d2: the dual of s2
\item bs1: each cube subdivided like the period of bitruncated cubic honeycomb (see Figure \ref{hcimages})
\item bs2: each cube subdivided like the double period of bitruncated cubic honeycomb (see Figure \ref{hcimages})
\end{itemize}

\end{document}